\documentclass[11pt]{article}

\usepackage{amsmath}
\usepackage[showonlyrefs]{mathtools}
\usepackage{mathstyle}
\usepackage{empheq}
\usepackage{amstext,amssymb,amsfonts}
\usepackage{fullpage}
\usepackage{nicefrac}
\usepackage{bm}
\usepackage{mdwlist}

\usepackage{textcomp,setspace}
\usepackage{todonotes}

\usepackage{nameref}
\usepackage[pagebackref,colorlinks,linkcolor=blue,filecolor = blue,citecolor = blue, urlcolor = blue, hyperfootnotes=false]{hyperref}
\usepackage{cleveref}
\renewcommand{\cref}{\Cref}

\usepackage{amsthm}
\usepackage{thmtools}

\makeatletter
\def\thmt@refnamewithcomma #1#2#3,#4,#5\@nil{%
  \@xa\def\csname\thmt@envname #1utorefname\endcsname{#3}%
  \ifcsname #2refname\endcsname
    \csname #2refname\expandafter\endcsname\expandafter{\thmt@envname}{#3}{#4}%
  \fi
}
\makeatother

\declaretheorem[within=section, name=Theorem, Refname=Theorem]{thm}
\crefname{thm}{Theorem}{Theorems}
\declaretheorem[within=section, name=Conjecture, Refname=Conjecture]{conjecture}
\crefname{conjecture}{Conjecture}{Conjectures}
\declaretheorem[sibling=thm, name=Claim, Refname=Claim]{claim}
\crefname{claim}{Claim}{Claims}

\declaretheorem[sibling=thm, name=Lemma, Refname=Lemma]{lem}
\crefname{lem}{Lemma}{Lemmas}
\declaretheorem[sibling=thm, name=Definition, Refname=Definition]{define}
\crefname{define}{Definition}{Definitions}

\declaretheorem[sibling=thm, name=Corollary, Refname=Corollary]{cor}
\crefname{cor}{Corollary}{Corollaries}

\declaretheorem[sibling=thm, name=Proposition, Refname=Proposition]{prop}
\crefname{prop}{Proposition}{Proposit	ions}
\declaretheorem[sibling=thm, name=Remark, Refname=Remark]{remark}
\crefname{remark}{Remark}{Remarks}

\renewcommand{\leq}{\leqslant}
\renewcommand{\geq}{\geqslant}
\renewcommand{\ge}{\geqslant}
\renewcommand{\le}{\leqslant}

\newcommand{\remove}[1]{}

\newcommand{\poly}{{\rm poly}}
\newcommand{\Tr}{{\rm Tr}}

\newcommand{\eps}{{\varepsilon}}
\renewcommand{\l}{\left}
\renewcommand{\r}{\right}

\newcommand{\de}{{\delta}}

\newcommand{\comments}[1]{}
\newcommand{\rank}{\textnormal{rank}}
\newcommand{\depth}{\textnormal{depth}}

\renewcommand{\deg}{\textnormal{deg}}
\renewcommand{\dim}{\textnormal{dim}}
\newcommand{\LDR}{\textnormal{LDR}}

\newcommand{\dist}{\textnormal{dist}}
\newcommand{\bias}{\textnormal{bias}}
\newcommand{\calF}{\mathcal{F}}
\newcommand{\calK}{\mathcal{K}}

\newcommand{\RM}{\textnormal{RM}}

\newcommand{\K}{\mathbb{K}}

\newcommand{\set}[1]{\{ #1 \}}
\newcommand{\zo}{\{0,1\}}

\def\F{{\mathbb{F}}}

\newcommand{\R}{\mathbb{R}}
\newcommand{\N}{\mathbb{N}}
\newcommand{\T}{\mathbb{T}}
\DeclareMathOperator*{\E}{\mathbb{E}}
\newcommand{\Z}{\mathbb{Z}}
\newcommand{\U}{\mathbb{U}}
\newcommand{\De}{\Delta}

\renewcommand{\P}{\mathcal{P}}
\newcommand{\calL}{\mathcal{L}}

\newcommand{\calH}{\mathcal{H}}
\newcommand{\calA}{\mathcal{A}}

\newcommand{\B}{\mathcal{B}}

\newcommand{\cC}{\mathcal{C}}
\newcommand{\C}{\mathbb{C}}

\renewcommand{\Pr}{\mathbf{Pr}}

\renewcommand{\char}{\textnormal{char}}

\def\draft{0}   

\ifnum\draft=1 
    \def\ShowAuthNotes{1}
\else
    \def\ShowAuthNotes{0}
\fi

\makeatletter
\renewcommand{\todo}[2][]{%
    \@todo[caption={#2}, #1]{\begin{spacing}{0.5}#2\end{spacing}}%
} 
\makeatother 
\newcounter{mynotes}
\ifnum\ShowAuthNotes=0
\newcommand{\mnote}[2]{\addtocounter{mynotes}{1}{{}}\todo[color=blue!20!white]{[\arabic{mynotes}] \tiny  {{#1}: {\sf {#2}}}}}
\else
\newcommand{\mnote}[2]{}
\fi

\newcommand{\ArnabNote}{\mnote{Arnab}}

\newcommand{\wt}{\mathsf{wt}}

\begin{document}
\title{Using higher-order Fourier analysis over general fields}

\author{
Arnab Bhattacharyya\thanks{Supported in part by a DST Ramanujan Fellowship.}\\
Department of Computer Science \& Automation\\
Indian Institute of Science\\
\texttt{arnabb@csa.iisc.ernet.in}
\and
Abhishek Bhowmick\thanks{Research supported in part by NSF Grant CCF-1218723.}\\
Department of Computer Science\\
The University of Texas at Austin\\
\texttt{bhowmick@cs.utexas.edu}
}

\maketitle

\begin{abstract}
Higher-order Fourier analysis, developed over prime fields, has been recently used in different areas of computer science, including list decoding, algorithmic decomposition and testing. We extend the tools of higher-order Fourier analysis to analyze functions over general fields. Using these new tools, we revisit the results in the above areas.
\begin{enumerate*}
\item[(i)]
For any fixed finite field $\K$, we show that the list decoding radius of the generalized Reed Muller code over $\K$ equals the minimum distance of the code. Previously, this had been proved over prime fields \cite{BL14} and for the case when $|\K|-1$ divides the order of the code \cite{GKZ08}.
\item[(ii)]
For any fixed finite field $\K$, we give a polynomial time algorithm to decide whether a given polynomial $P: \K^n \to \K$ can be decomposed as a particular composition of lesser degree polynomials. This had been previously established over prime fields \cite{B14, BHT15}.
\item[(iii)]
For any fixed finite field $\K$, we prove that all locally characterized affine-invariant properties of functions $f: \K^n \to \K$ are testable with one-sided error.  The same result was known when $\K$ is prime \cite{BFHHL13} and when the property is linear \cite{KS08}. Moreover, we show that for any fixed finite field $\F$, an affine-invariant property of functions $f: \K^n \to \F$, where $\K$ is a growing field extension over $\F$, is testable if it is locally characterized by constraints of bounded weight.
\end{enumerate*}
\end{abstract}

\section{Introduction}

Fourier analysis over finite groups has played a central role in the development of theoretical computer science. Examples of its applications are everywhere: analysis of random walks on graphs \cite{CDG87}, fast integer multiplication algorithms \cite{SS71}, learning algorithms \cite{KM93}, the Kahn-Kalai-Linial theorem \cite{KKL88}, derandomization \cite{NN93}, tight inapproximability results using probabilistically checkable proofs \cite{Hastad01}, social choice theory \cite{MOO10},  and coding theory \cite{NS05}. See the surveys of De Wolf \cite{deWolf} and \v{S}tefankovi\v{c} \cite{Ste00}. 

Higher-order Fourier analysis is a recent generalization of some aspects of Fourier analysis. Consider functions over the integers $\Z$. While classical Fourier analysis over $\Z$ studies correlations of functions with linear phases $e^{i \theta n}$, higher-order Fourier analysis over $\Z$ analyzes the correlation of functions with polynomial phases such as $e^{i \theta n^2}$. The modern\footnote{In retrospect, Weyl's results on equidistribution of polynomial phases \cite{Weyl14} laid the foundations of this theory.} work on higher-order Fourier analysis over $\Z$ began with the spectacular proof by Gowers of Szemer\'edi's theorem \cite{Gow98, Gow01}, where the {\em Gowers norm} was introduced, and with the ergodic theory work of Host and Kra \cite{HK05}. Subsequently, Green, Tao and Ziegler through several works \cite{GT05, GT06, GTZ11, GTZ12} largely completed the research program of understanding the relationships between different aspects of the theory over $\Z$. This work was applied to solve several longstanding open problems in additive number theory, including the celebrated result showing the existence of arbitarily long arithmetic progressions in the primes \cite{GT06}.  The book \cite{Tao12} by Tao on the subject surveys the current state of knowledge.

In an influential article \cite{Green05b}, Green popularized the idea that it is useful to rephrase the problems arising in additive number theory into problems on vector spaces over fixed finite fields. The motivation was that many of the techniques in higher-order Fourier analysis over $\Z$ simplify over finite fields, because of the presence of subspaces and of algebraic notions such as orthogonality and linear independence. However, it was soon realized that these questions over finite fields are also intrinsically interesting because of their connections to theoretical computer science. In particular, the Gowers norm for functions on $\F^n$ for a finite prime field $\F$ is directly related to low-degree testing, a problem intensely studied by computer scientists since the early 90's. 

Thanks to the sequence of works \cite{GT09,KL08,TZ10,BTZ10,TZ12}, the apparatus of higher-order Fourier analysis over $\F^n$ for any fixed prime order field $\F$ is also now largely complete. The theory has subsequently found several interesting applications in computer science that we detail below and has become part of the mainstream theorist toolkit. However, in all of these applications, the finite field in consideration was restricted to be a field of {\em prime} order, while the problems themselves are interesting over general finite fields. In this work, we show how the techniques of higher-order Fourier analysis continue to apply even when the underlying field is a non-trivial extension of a prime order field. 

\subsection{Applications}
In this section, we describe three different problems involving a finite field $\K$, which previously had been solved only when $|\K|$ was prime but which we can now solve for arbitrary finite $\K$.

Throughout, let $\F$ be a fixed prime order field, and let $\K$ be a finite field that extends $\F$. Let $q = |\K|$, $p = |\F|$ and $q = p^r$ for $r>0$. 

\subsubsection{List-decoding Reed-Muller codes}
The notion of \emph{list decoding} was introduced by Elias \cite{Elias} and Wozencraft \cite{Woz} to decode \emph{error correcting codes} beyond half the minimum distance. The goal of a list decoding algorithm is to produce all the codewords within a specified distance from the received word. At the same time one has to find the right radius for which the number of such codewords is small, otherwise there is no hope for the algorithm to be efficient. After the seminal results of Goldreich and Levin \cite{GL} and Sudan \cite{Sudan} which gave list decoding algorithms for the Hadamard code and the Reed-Solomon code respectively, there has been tremendous progress in designing list decodable codes. See the survey by Guruswami \cite{Venkat:book, Venkat:thesis} and Sudan \cite{Sudan:survey}.

List decoding has applications in many areas of computer science including hardness amplification in complexity theory \cite{STV, luca-xor}, derandomization \cite{Vad}, construction of hard core predicates from one way functions \cite{GL, AGS}, construction of extractors and pseudorandom generators \cite{TZS, SU} and computational learning \cite{KM,Jackson}. However, the largest radius up to which list decoding is tractable is still a fundamental open problem even for well studied codes like Reed-Solomon (univariate polynomials) and Reed-Muller codes (multivariate polynomials). 
The goal of this work is to analyse Reed-Muller codes over small fields (possible non prime) and small degree.

Reed-Muller codes (RM codes) were discovered by Muller in 1954. Let $d \in \N$. The RM code $\RM_{\K}(n,d)$ is defined as follows. The message space consists of degree $\leq d$ polynomials in $n$ variables over $\K$ and the codewords are evaluation of these polynomials on $\K^n$. Let $\de_q(d)$ denote the normalized distance of $\RM_{\K}(n,d)$. Let $d=a(q-1)+b$ where $0 \leq b<q-1$. We have
$$
\de_{\K}(d)=\frac{1}{q^a}\l(1-\frac{b}{q}\r).
$$

RM codes are one of the most well studied error correcting codes. Many applications in computer science involve low degree polynomials over small fields, namely RM codes. Given a received word $g:\K^n \rightarrow \K$ the objective is to output the list of codewords (e.g. low-degree polynomials) that lie within some distance of $g$. Typically we will be interested in regimes where list size is either independent of $n$ or polynomial in the block length $q^n$.

 Let $\P_{d}(\K^n)$ denote the class of degree $\leq d$ polynomials $f:\F^n \rightarrow \F$. Let $\dist$ denote the normalized Hamming distance.
For $\RM_{\K}(n,d)$, $\eta>0$, let
$$
\ell_{\F}(n,d,\eta):=\max_{g:\F^n \rightarrow \F}\l|\{f \in \P_d(\F^n):\dist(f,g) \leq \eta\}\r|.
$$
Let $\LDR_{\K}(n,d)$ (short for \emph{list decoding radius}) be the maximum $\rho$ for which $\ell_{\K}(n,d,\rho-\eps)$ is upper bounded by a constant depending only on $\eps, |\K|,d$ for all $\eps>0$.

It is easy to see that $\LDR_{\K}(n,d)\leq \de_{\K}(d)$. The difficulty lies in proving a matching lower bound. We review some previous work next. The first breakthrough result was the celebrated work of Goldreich and Levin \cite{GL} who showed that in the setting of $d=1$ over $\F_2$ (Hadamard Codes) $\LDR_{\F_2}(n,1)=\de_{\F_2}(1)=1/2$. Later, Goldreich, Rubinfield and Sudan \cite{GRS} generalized the field to obtain $\LDR_{\K}(n,1)=\de_{\K}(1)=1-1/|\K|$. In the setting of $d<|\K|$, Sudan, Trevisan and Vadhan \cite{STV} showed that $\LDR_{\K}(n,d) \ge 1-\sqrt{2d/|\K|}$ improving previous work by Arora and Sudan \cite{AroraSudan}, Goldreich \emph{et al} \cite{GRS} and Pellikaan and Wu \cite{PellikaanWu}. Note that this falls short of the upper bound which is $\de_{\K}(d)$.

In 2008, Gopalan, Klivans and Zuckerman \cite{GKZ08} showed that $\LDR_{\F_2}(n,d)=\de_{\F_2}(d)$.They posed the following conjecture.
\begin{conjecture}[\cite{GKZ08}]\label{conj:1}For fixed $d$ and finite field $\K$, $\LDR_{\K}(n,d)=\de_{\K}(d)$.
\end{conjecture}

It is believed \cite{GKZ08, Gopalan10} that the hardest case is the setting of small $d$.
An important step in this direction was taken in \cite{Gopalan10} that considered quadratic polynomials and showed that $\LDR_{\K}(n,2) = \de_{\K}(2)$ for all fields $\K$ and thus proved the conjecture for $d=2$. Recently, Bhowmick and Lovett \cite{BL14} resolved the conjecture for prime $\K$.

Our main result for list decoding is a resolution of Conjecture~\ref{conj:1}.

\begin{thm}\label{thm:main}Let $\K$ be a finite field. Let $\eps>0$ and $d,n \in \N$. Then, $$\ell_{\K}(d,n,\de_{\K}(d)-\eps) \leq c_{|\K|,d,\eps}.$$ Thus, $$\LDR_{\K}(n,d)=\de_{\K}(d).$$
\end{thm}

\begin{remark}[Algorithmic Implications]Using the blackbox reduction of algorithmic list decoding to combinatorial list decoding in \cite{GKZ08} along with Theorem~\ref{thm:main}, for fixed finite fields, $d$ and $\eps>0$, we now have list decoding algorithms in both the global setting (running time polynomial in $|\K|^n$) and the local setting (running time polynomial in $n^d$).
\end{remark}

\subsubsection{Algorithmic polynomial decomposition}

Consider the following family of
properties of functions over a finite field $\K$.
\begin{define}
Given a positive integer $k$, a vector of positive integers $\bm{\Delta} = (\Delta_1, \Delta_2,
\dots, \Delta_k)$ and a function $\Gamma: \K^k \to \K$, we say that a
function $P: \K^n \to \K$ is {\em  $(k,\bm{\Delta},\Gamma)$-structured} if there exist polynomials
$P_1, P_2, \dots, P_k: \K^n \to \K$ with each $\deg(P_i) \leq \Delta_i$ such that
for all $x \in \K^n$,  $$P(x) = \Gamma(P_1(x), P_2(x), \dots,
P_k(x)).$$ The polynomials $P_1, \dots, P_k$ are said to form a
{\em $(k,\bm{\Delta}, \Gamma)$-decomposition}. 
\end{define}
For instance, an $n$-variate polynomial over the field
$\K$ of total degree $d$  factors nontrivially exactly when it is $(2, (d-1,d-1),
\mathsf{prod})$-structured where $\mathsf{prod}(a,b) = a\cdot
b$. We shall use the term {\em degree-structural property} to refer to a property
from the family of $(k,\bm{\Delta}, \Gamma)$-structured properties.

The problem here is, for arbitrary fixed $k, \K, \bm(\Delta), \Gamma$, given a polynomial, decide efficiently if it is degree structural and if yes, output the decomposition.
An efficient algorithm for the above would imply
a (deterministic) $\poly(n)$-time algorithm for factoring an $n$-variate polynomial
of degree $d$ over $\K$. Also, it implies a polynomial
time algorithm for deciding whether a $d$-dimensional tensor over $\K$
has rank at most $r$.
\remove{
As noted in \cite{B14}, these results on factoring and tensor rank are not
new, in the sense 
that there were already algorithms known for stronger versions of
these two problems. See \cite{KS09, GK85, Sud12}.}
Also, it would give polynomial time algorithms for a wide range of
problems not known to have non-trivial solutions previously, such as
whether a polynomial of degree $d$ can be expressed as $P_1\cdot P_2 +
P_3 \cdot P_4$ where each $P_1,P_2,P_3,P_4$ are of degree $d-1$ or
less.

This problem was solved for prime $\K$, satisfying $d<|\F|$ by Bhattacharyya~\cite{B14} and later for all $d$ and prime $|\K|$ by Bhattacharyya, Hatami and Tulsiani \cite{BHT15}.

Our main result in this line of work establishes this for all fixed finite fields.
\begin{thm}\label{main}
For every finite field $\K$, positive integers $k$ and $d$, every vector of positive integers
$\bm{\Delta} = (\Delta_1, \Delta_2, \dots, \Delta_k)$ and every
function $\Gamma: \K^k \to \K$, there is a deterministic algorithm
$\mathcal{A}_{\K, d, k,\bm{\Delta},\Gamma} $ that takes as input a polynomial
$P: \K^n \to \K$ of degree $d$ that runs in time polynomial in $n$, and outputs
a $(k,\bm{\Delta}, \Gamma)$-decomposition of $P$ if one exists while
otherwise returning $\mathsf{NO}$. 
\end{thm}

\subsubsection{Testing affine-invariant properties}
The goal of property testing, as initiated by \cite{BLR93, BFL91} and defined formally by \cite{RS96, GGR98}, is to devise algorithms that query their input a very small number of times while correctly deciding whether the input satisfies a given property or is ``far'' from satisfying it. A property is called {\em testable} if the query complexity can be made independent of the size of the input.

More precisely, we use the following definitions. Let $[R]$ denote the set $\set{1, \dots, R}$. Given a property $\P$ of functions in $\{\K^n \to [R] \ | \ n \in \Z_{\ge 0}\}$, we say that $f : \K^n \to [R]$ is {\em $\eps$-far} from $\P$ if
$$\min_{g \in \P} \Pr_{x \in \K^n}[f(x) \neq g(x)] > \eps,$$
and we say that it is {\em $\eps$-close} otherwise.

\begin{define}[Testability]\label{testable}
A property $\P$ is said to be {\em testable} (with one-sided error)
if there are functions $q: (0,1) \to \Z_{> 0}$, $\delta: (0,1) \to (0,1)$,
and an algorithm $T$ that, given as input a parameter $\eps > 0$ and oracle
access to a function $f: \K^n \to [R]$, makes at most $q(\eps)$
queries to the oracle for $f$, always accepts if $f \in \P$ and
rejects with probability at least $\delta(\eps)$ if $f$ is $\eps$-far
from $\P$. If, furthermore, $q$ is a constant function, then $\P$ is
said to be {\em proximity-obliviously testable (PO testable)}.
\end{define}

The term proximity-oblivious testing is coined by Goldreich and Ron in~\cite{GR11}.
As an example of a testable (in fact, PO testable) property, let us recall the famous result by Blum, Luby and Rubinfeld
\cite{BLR93} which initiated this line of research. They showed
that linearity of a function $f: \K^n \to \K$ is testable by a test
which makes $3$ queries. This test accepts if $f$ is linear and rejects with
probability $\Omega(\eps)$ if $f$ is $\eps$-far from linear.

Linearity, in addition to being testable, is also an example of a
{\em linear-invariant} property. We say that a property $\P
\subseteq \{\K^n \to [R]\}$ is linear-invariant if it is the case
that for any $f \in \P$ and for any $\K$-linear transformation $L: \K^n \to \K^n$, it holds that $f\circ L \in \P$.  Similarly, an
{\em affine-invariant} property is closed under composition with
affine transformations $A: \K^n \to \K^n$ (an affine transformation
$A$ is of the form $L+c$ where $L$ is $\K$-linear and $c\in \K$).
The property of a function $f: \K^n \to \K$ being affine is testable
by a simple reduction to \cite{BLR93}, and is itself affine-invariant.
 Other well-studied
examples of affine-invariant (and hence, linear-invariant) properties
include Reed-Muller codes 
\cite{BFL91, BFLS,FGLSS,RS96,AKKLR05}
and Fourier sparsity
\cite{GOSSW}.
In fact, affine invariance seems to be a common feature of most interesting
properties that one would classify as ``algebraic''.  Kaufman and
Sudan in
\cite{KS08}
made
explicit note of this phenomenon and initiated a general study of the testability of
affine-invariant properties (see also~\cite{GK11}). 

Our main theorem for testing is a very general positive result:
\begin{thm}[Main testing result]\label{maintest}
Let $\P \subseteq \{\K^n \to [R]\}$ be an affine-invariant property that is $t,w$-lightly locally characterized, where $t$, $R$, $w$, and $\char(\K)$ are fixed positive integers. Then, $\P$ is PO testable with $t$ queries.
\end{thm}
We are yet to define several terms in the above claim, but as we will see, the weight restriction is trivial when the field size is bounded. This yields the following characterization.
\begin{thm}[Testing result  for fixed fields]\label{testfixed}
Let $\P \subseteq \{\K^n \to [R]\}$ be an affine-invariant property, where $R \in \Z^+$ and field $\K$ are fixed. Then, $\P$ is PO testable with $t$ queries if and only if $\P$ is $t$-locally characterized.
\end{thm}

Previously, \cite{BFHHL13} (building on \cite{BCSX09, BGS10, BFL12}) proved \cref{maintest} in the case that $\K$ is of fixed prime order using higher-order Fourier analytic techniques. We note that other recent results on $2$-sided testability of affine-invariant properties over fixed prime-order fields \cite{HL13, Yoshida14} can also be similarly extended to non-prime fields but we omit their description here.

\paragraph{Local Characterizations}
For a PO testable property $\P \subset \{\K^n \to [R]\}$ of query complexity $t$, if a function $f: \K^n \to [R]$ does not satisfy $\P$, then by \cref{testable}, the tester rejects $f$ with positive probability. Since the test always accepts functions with the property, there must be $t$ points $a_1, \dots, a_t \in \K^n$ that form a witness for non-membership in $\P$. These are the queries that cause the tester to reject. Thus, denoting $\sigma = (f(a_1), \dots, f(a_t)) \in [R]^t$, we say that $\mathcal{C} = (a_1, a_2, \dots, a_t; \sigma)$ forms a {\em $t$-local constraint} for $\P$. This means that whenever the constraint is violated by a function $g$, i.e., $(g(a_1), \dots, g(a_t)) = \sigma$, we know that $g$ is not in $\P$.   A property $\P$ is {\em $t$-locally characterized} if there exists a collection of $t$-local constraints $\cC_1, \dots, \cC_m$ such that $g \in \P$ if and only if none of the constraints  $\cC_1, \dots, \cC_m$ are violated.  It follows from the above discussion that if $\P$ is PO testable with $q$ queries, then $\P$ is
$t$-locally characterized.

For an affine-invariant property, constraints can be defined in terms of affine forms, since the affine orbit of a constraint is also a constraint. So, we can describe each $t$-local constraint $\cC$ as $(A_1, \dots, A_t; \sigma)$, where for every $i \in [t]$, $A_i(X_1, \dots, X_t) = X_1 + \sum_{j=2}^t c_{i,j} X_j$ for some $c_{i,j} \in \K$ is an affine form over $\K$. We define the {\em weight} $\mathsf{wt}$ of an element $c \in \K$ as $\sum_{k=1}^r |c_k|$, where $c$ is viewed as an $r$-dimensional vector $(c_1, \dots, c_r)$ with each $c_i$ in the base prime field\footnote{If $x \in \F$, $|x|$ is the obvious element of $\{0, 1, \dots, |\F|-1\}$.} $\F$ with respect to a fixed arbitrary basis. The {\em weight of an affine form} $A_i$ to be $\sum_{j=2}^m \mathsf{wt}(c_{i,j})$ for $c_{i,j}$ as above. A constraint is said to be of weight $w$ if all its affine forms are of weight at most $w$, and a property $\P$ is said to be $t,w$-lightly localyl characterized if there exist $t$-local constraints $\cC_1, \dots, \cC_m$, each of weight at most $w$ that characterize $\P$. 

\cref{maintest} asserts that if $\P$ has a light local characterization, then it is testable. There can exist many local characterizations of a property, and for the theorem to apply, it is only necessary that one such characterization be of bounded weight. Moreover, we can choose the basis with which to describe $\K$ over $\F$. On the other hand, some restriction in addition to local characterization is needed, as Ben-Sasson et al.~\cite{BMSS11} show that there exist affine-invariant locally characterized properties of functions $f: \F_{2^n} \to \F_2$ that require super-constant query complexity to test.

Another interesting observation is that if a property has a local characterization of bounded weight, then it has a local {\em single orbit characterization}, in the language of \cite{KS08}. For linear\footnote{These are properties of functions $f: \K^n \to \F$, where $\F$ is a subfield of $\K$, for which $f, g \in \P$ implies $\alpha f + \beta g \in \P$ for any $\alpha, \beta \in \F$.} affine-invariant properties, \cite{KS08} shows that any local single orbit characterized property is testable. Hence, our result is weaker than \cite{KS08} in this aspect, though our \cref{maintest} allows non-linear properties. It is an interesting open question as to whether dual-BCH codes and, more generally, sparse affine-invariant codes that were shown to be locally single orbit characterized in \cite{KL05} and \cite{GKS12} respectively also have local characterizations of bounded weight. It is also an open problem to describe a testable property $\P \subseteq \{\F_{2^n} \to \F_2\}$ that does not have a local characterization of bounded weight.

\subsection{Our Techniques}

\subsubsection{New Ingredients}
Our starting point is the observation that $\K$ is an $r$-dimensional vector space over $\F$. Thus, we can view a function $Q: \K^n \to \K$ as determined by a collection of functions $P_1, \dots, P_r: \K^n \to \F$ where $\K^n$ is viewed as $\F^{rn}$. In view of this, we define the notion of an {\em additive polynomial}. A function\footnote{To deal with low characteristics, we will actually use a slightly general definition valid for functions mapping to the torus $\R/\Z$.} $P: \K^n \to \F$ is said to have {\em additive degree $d$} if for all $h_1, \dots, h_{d+1} \in \K^n$, $D_{h_1}\cdots D_{h_{d+1}}P \equiv 0$, where $(D_hP)(x) = P(x+h)-P(x)$. Additive polynomials are exactly the non-classical polynomials of \cite{TZ12} when the domain is $\F^{rn}$. Moreover, if $Q: \K^n \to \K$ has degree $d$ (in the usual sense of having a monomial with degree $d$), then $\Tr(\alpha Q)$ has additive degree $\leq d$ for any $\alpha \in \K$ where $\Tr: \K \to \F$ denotes the trace function.

Therefore, we can directly write any polynomial $P: \K^n \to \K$ in terms of additive polynomials and then import all of the results shown in \cite{TZ12} for non-classical polynomials to our setting! Unfortunately, we are not done. The reason is that our applications require, in addition to additive structure, some of the multiplicative structure of $\K$, which is lost when we view $\K$ as $\F^r$.

To see why, recall the question of testing affine-invariant properties. When $\K$ is of bounded order, we can view any one-sided test as examining the restriction of the input function on a random $K$-dimensional affine subspace of $\K^n$, for some constant integer $K$. In other words, the test will evaluate the input function at elements of the set $H=\{x + \sum_{i=1}^K a_i y_i: a_1, \dots, a_K \in \K\}$ for some $x, y_1, \dots, y_K \in \K$. Clearly, $H$ is not an affine subspace of $\F^{rn}$. An important component of the higher-order Fourier analytic approach is to show that any ``sufficiently pseudorandom'' collection of polynomials is equidistributed on $H$, and the proof of this fact in \cite{BFHHL13} crucially uses that $H$ is a subspace of a vector space over a prime field. In our work, we show a strong equidistribution theorem (\cref{unifsubspace}) that holds when $H$ is an affine subspace of $\K^n$. 

A different place where multiplicative structure rears its head is a key {\em Degree Preserving Lemma} of \cite{BFHHL13}. Informally, it states that if $P_1, \dots, P_C$ form a ``sufficiently pseudorandom'' collection of polynomials and $F(x) = \Gamma(P_1(x), \dots, P_C(x))$ is a polynomial of degree $d$ where $\Gamma$ is an arbitrary composition function, then for any other collection of polynomials $Q_1, \dots, Q_C$ where $\deg(Q_i) \leq \deg(P_i)$ for every $i$, $G(x) = \Gamma(Q_1(x), \dots, Q_C(x))$ also has degree $\leq d$. The lemma is crucially used for the analysis of the Reed-Muller list decoding bound in \cite{BL14} and the polynomial decomposition algorithm in \cite{B14, BHT15}. Its proof goes via showing that if all $(d+1)$ iterated derivatives of $F: \K^n \to \K$ vanish, then so must all $(d+1)$ iterated derivatives of $G: \K^n \to \K$. However, when $|\K|$ is non-prime, all $(d+1)$ iterated derivates of a function $G: \K^n \to \K$ may vanish without the degree being $\leq d$; consider for example the polynomial $x^p$ which vanishes after only $2$ derivatives.

We resolve this issue by giving a different and more transparent proof of the Degree Preserving Lemma, which actually holds in a much more general setting (\cref{proppreserve}). Using the above notation, we prove that if $F: \K^n \to \K$ satisfies some locally characterized property $\P$, then $G: \K^n \to \K$ does also. Since due to a work of Kaufman and Ron \cite{KR06}, we know that degree is locally characterized, our desired result follows. Our new proof uses our strong equidistribution theorem on affine subspaces of $\K^n$.

An interesting point to note is that both the equidistribution theorem and the degree preserving lemma work only assuming that the field characteristic is constant and that the involved affine constraints are of bounded weight, without any assumption on the field size.

 \subsubsection{Reed-Muller codes}
 For a received word $g:\K^n \to \K$ our goal is to upper bound $\l|\{f \in \P_d:\dist(f,g) \leq \eta\}\r|$, where $\eta=\de_{\K}(d)-\eps$ for some $\eta>0$ and $\P_d$ is the class $\{Q: \K^n \to \K: \deg(Q) \leq d\}$. The proof technique is similar in structure as \cite{BL14}. We apply the weak regularity lemma (Corollary~\ref{cor:pseudorandom}) to the received word $g:\K^n \to \K$ and reduce the problem to a structured word $g':\K^n \to \K$. More specifically, whenever $\dist(f,g) \le \eta$, we have $\dist(f,g') \le \eta+\eps/2$.  From here, we first express each function $f:\K^n \to \K$ as a linear combination of functions $f':\K^n \to \F$. It can be then shown that the analysis in \cite{BL14} works for functions $f':\K^n \to \F$. A naive recombination of the $f':\K^n \to \F$ to $f:\K^n \to \K$ gives us useful bounds only when $d<\char(|\F|)$. To circumvent this problem, we use our improved degree preserving theorem. This is crucial to our analysis as the technique of \cite{BL14} can be used only to analyze the additive degree of polynomials which is not enough for the argument to work for arbitrary $d$ and $|\K|$.

\subsubsection{Polynomial decomposition}
The algorithm and its analysis follows the lines of \cite{B14, BHT15}. Given a polynomial $P: \K^n \to \K$ (where $|\K|$ is bounded), we consider the collection of additive polynomials $\{\Tr(\alpha_1 P), \dots, \Tr(\alpha_r P)\}$ where $\alpha_1, \dots, \alpha_r \in \K$ are linearly independent. We regularize this collection into a pseudorandom additive polynomial factor and set one variable to $0$ such that the degrees of the polynomials do not change. We then recursively solve the problem on $n-1$ variables and then apply a lifting procedure to get a decomposition for the original problem. A naive analysis of the lifting procedure over non-prime fields requires that $\deg(P) < \char(\F)$. In order to get around this, we use our improved degree preserving theorem which applies for arbitrary degrees.

 \subsubsection{Testing affine-invariant properties}
 Suppose $\P \subseteq \{\K^n \to [R]\}$ is a locally characterized affine-invariant property (where $R$ and $\char(\K)$ are bounded but $n|\K|$ is growing). Our proof follows the lines of \cite{BGS10, BFL12, BFHHL13}. Suppose $f$ is far from $\P$. We first identify a low-rank function close to $f$ in an appropriate Gowers norm which also contains the violation that $f$ contains. Here, low rank is with respect to a collection $\B$ of {\em additive} polynomials.  We then investigate the distribution of $\B$ on the affine constraint that $f$ violates. Since these are affine with respect to $\K^n$, we need to use our strong equidistribution theorem. The rest of the proof proceeds along the same lines as \cite{BFHHL13}.
 
 Because the proof of \cref{maintest} is very analogous to that in \cite{BFHHL13} (except for the use of additive polynomials and the new equidistribution theorem) and requires significant additional notation, we omit it here.

\section{Preliminaries}
Let $\N$ denote the set of positive integers. For $n \in \N$, let $[n]:=\{1,2,\ldots , n\}$. We use $y=x \pm \eps$ to denote $y \in [x-\eps, x+\eps]$. For $n \in \N$, and $x,y \in \C^n$, let $\langle x,y \rangle:=\sum_{i=1}^n x_i\overline{y_i}$ where $\overline{a}$ is the conjugate of $a$. Let $\|x\|_2:=\sqrt{\langle x,x \rangle}$.

Let $\T$ denote the torus $\R/\Z$. This is an abelian group under addition.  Let $e: \T \to \C$ be the function $e(x) = e^{2\pi ix}$. For an integer $k \geq 0$, let $\U_{k}:=\frac{1}{p^k}\Z/\Z$. Note that $\U_k$ is a subgroup of $\T$. Let $\iota:\F \rightarrow \U_1$ be the bijection $\iota(a)=\frac{|a|}{p} \pmod 1$.

Fix a prime field $\F=\F_p$, and let $\K = \F_q$ where $q=p^r$ for a positive integer $r$. 
We denote by $\Tr:\K \to \F$ the trace function:
$$\Tr(x) = x + x^p + x^{p^2} + \cdots + x^{p^{r-1}}$$
Recall that $\{x \to \Tr(ax) : a \in \K\}$ is in bijection with the set of all linear maps from $\K$ to $\F$. Also, we use $|\cdot|$ to denote the obvious map from $\F$ to $\{0,1,\dots,p-1\}$. 
We will need the following useful fact.
\begin{prop}[Dual basis]\label{dual}
For any $r$ linearly independent elements $\alpha_1, \dots, \alpha_r \in \K$, there exist $ \beta_1, \beta_2, \dots, \beta_r$ in $\K$ such that any $x \in \K$ equals $\sum_{i=1}^r \beta_i \Tr(\alpha_ix)$.
\end{prop}

Given a basis, i.e. collection of $r$ linearly independent field elements, $\bm{\alpha}=(\alpha_1, \dots, \alpha_r)$, we define $\mathsf{wt}_{\bm{\alpha}}: \K \to \Z$ to be $\wt_{\bm{\alpha}}(c) = \sum_{i=1}^r |\Tr(\alpha_i c)|$. 

\remove{
For a finite set $X$ and $n \in \N$, with $f:X\rightarrow \C^n$, we write $\E_{x}f(x)$ to denote $\frac{1}{|X|}\sum_{x \in X}f(x)$. We define $||f||_2:=\sqrt{\E_x ||f(x)||_2^2}$. If $g: X \rightarrow \C^n$, we have $\langle f,g \rangle  :=  \E_x \langle f(x),g(x) \rangle$. Let $Y$ be a finite set. Let $P(Y):=\{f:Y \rightarrow \R_{\geq 0}  :\sum_{y \in Y}f(y)=1\}$ denote the probability simplex on $Y$. We shall write randomized functions by mapping them to the simplex. Thus, for $f,g:X \rightarrow P(Y)$ we define $$\Pr_{x}[f(x)=g(x)]:=\E_x\langle f(x),g(x)\rangle.$$
If $f:X \rightarrow Y$ is a deterministic function, then we embed $Y$ into $P(Y)$ in the obvious way, and consider $f:X \to P(Y)$ with $f(x)_{y}=1$ if $f(x)=y$ when viewed as a function to $Y$, and $f(x)_{y'}=0$ for all $y' \in Y \setminus \{y\}$.
}

\subsection{Affine forms and constraints}

A {\em linear form on $k$ variables} is a vector $L = (w_1, w_2, \dots, w_k) \in \K^k$ that is interpreted as a function from $(\K^n)^k$ to $\K^n$ via the map $(x_1, \dots, x_k) \mapsto w_1 x_1 + w_2 x_2 + \cdots + w_k x_k$. A linear form $L = (w_1, w_2, \dots, w_k)$ is said to be {\em affine} if $w_1 =1$. From now, linear forms will always be assumed to be affine. Given a basis $\bm{\alpha} = (\alpha_1, \dots, \alpha_r)$, we define $\mathsf{wt}_{\bm{\alpha}}$ of a linear form $L = (w_1, \dots, w_k)$ to be $\sum_{i=2}^k \mathsf{wt}_{\bm{\alpha}}(w_i)$. 

We specify a partial order $\preceq$ among affine forms, with respect to a basis $\bm{\alpha}=(\alpha_1, \dots, \alpha_r)$. We say $(w_1, \dots, w_k) \preceq_{\bm{\alpha}} (w_1',\dots, w_k')$ if $|\Tr(\alpha_j w_i)| \leq |\Tr(\alpha_j w_i')|$ for all $i \in [k], j \in [r]$. 

\begin{define}[Affine constraints]\label{defaffine}
An {\em affine constraint of size $m$ on $k$ variables} is a tuple $A = (L_1, \dots, L_m)$ of $m$ affine forms $L_1, \dots, L_m$ over
$\F$ on $k$ variables, where:
$L_1(x_1, \dots, x_k) = x_1$.
Moreover, it is said to be {\em weight-closed} if there exists a basis $\bm{\alpha}=(\alpha_1, \dots, \alpha_r)$ such that for any affine form $L$ belonging to $A$, if $L' \preceq_{\bm{\alpha}} L$, then $L'$ also belongs to $A$.
\end{define}

Observe that a weight-closed affine constraint is of bounded size if and only if all its affine forms are of bounded weight with respect to some $\bm{\alpha}$.

\subsection{Polynomials, Degrees and Derivatives}
A function $P: \K^n \to \K$ is a {\em polynomial of degree $d$} if for all $d_1, \dots, d_n \geq 0$ such that $\sum_i d_i \leq d$, there exists $c_{d_1, \dots, d_n} \in \K$ such that:
$$P(x_1, \dots, x_n) = \sum_{\substack{d_1, \dots, d_n \in \Z^+: \\  d_1 + \cdots + d_n \leq d}} c_{d_1, \dots, d_n} x_1^{d_1} x_2^{d_2} \cdots x_n^{d_n}$$

We use the notion of {\em additive degree} for functions mapping to $\T$.  Given a function $f:\K^n \to \T$, its {\em additive derivative in direction $h \in \K^n$} is $D_h f:\K^n \to \T$, given by
$$
D_h f(x) = f(x+h) - f(x).
$$
\begin{define}[Additive Polynomials]
A function $P: \K^n \to \T$ is a {\em polynomial of additive degree $d$} if for all $x, h_1, h_2, \dots, h_{d+1} \in \K^n$, we have 
\begin{equation}\label{defadd}
D_{h_1} D_{h_2} \cdots D_{h_{d+1}} P(x) = 0.
\end{equation} A function of bounded additive degree is called an {\em additive polynomial}. 	
\end{define}
For functions $P$ mapping to $\T$, $\deg(P)$ denotes its additive degree.
Note that we can interpret $P: \K^n \to \T$ as a function $P': \F^{nr} \to \T$ with the same additive degree by setting $P(x_1, \dots, x_n) = P'(\Tr(\alpha_1 x_1), \dots, \Tr(\alpha_1 x_1), \dots, \Tr(\alpha_1 x_n), \dots, \Tr(\alpha_1 x_n))$, using Proposition \ref{dual}. By this identification, additive polynomials are exactly the same as the non-classical polynomials introduced by Tao and Ziegler \cite{TZ12}. As a consequence, we have the following:
\begin{lem}[Lemma 1.7 of \cite{TZ12}]\label{explicit}
$P: \K^n \to \T$ is a polynomial of additive degree $d$ if and only if it can be written in the form:
$$
P(x_1, \dots, x_n) = \alpha + \sum_{k \geq 0} \sum_{\substack{0 \leq d_{i,j} < p~ \forall i \in [n], j \in [r]: \\ 0 < \sum_{i=1}^n \sum_{j=1}^r d_{i,j} \leq d-k(p-1)}} \frac{c_{d_{1,1}, \dots, d_{n,r}, k} \prod_{i=1}^n \prod_{j=1}^r |\Tr(\alpha_j x_i)|^{d_{i,j}}}{p^{k+1}} \pmod 1
$$
where $\alpha \in \T$ and $c_{d_{1,1}, \dots, d_{n,r}} \in \{0,1,\dots, p-1\}$ are uniquely determined. The maximum $k$ for which there is a nonzero $c_{d_{1,1}, \dots, d_{n,r}, k}$ is the {\em depth} of $P$. Note that $\text{depth}(P) \leq \left\lfloor \frac{d-1}{p-1}\right\rfloor$ and  that $P$ takes on at most $p^{\text{depth}(P)+1}$ distinct values.
\end{lem}

For a function $f: \K^n \to \C$, define the {\em multiplicative derivative in direction $h \in \K^n$} to be be 
$$\Delta_hf(x) = f(x+h) \cdot \overline{f(x)}.$$

\subsection{Locally Characterized Properties}
As described in the introduction, by a locally characterized property, we informally mean a property for which non-membership can be certified by a finite sized witness. Specifically for affine-invariant properties, we define:
\begin{define}[Locally characterized properties]\label{localchar}
\
\begin{itemize}
\item
An {\em induced affine constraint of size $m$ on $\ell$ variables} is
a pair $(A,\sigma)$ where $A$ is an affine constraint of size $m$ on
$\ell$ variables and $\sigma \in [R]^m$.
\item
Given such an induced affine constraint $(A,\sigma)$, a function $f:
\K^n \to [R]$ is said to be {\em $(A,\sigma)$-free} if there exist no
$x_1, \dots, x_\ell \in \K^n$ such that $(f(L_1(x_1, \dots, x_\ell)),
\dots, f(L_m(x_1,\dots,x_\ell))) = \sigma$. On the other hand, if such
$x_1, \dots, x_\ell$ exist, we say that {\em $f$ induces $(A,\sigma)$ at
  $x_1, \dots, x_\ell$}.
\item
Given a (possibly infinite) collection $\mathcal{A} = \{(A^1,\sigma^1),
(A^2, \sigma^2), \dots, (A^i,\sigma^i),\dots\}$ of induced affine constraints, a function $f: \K^n
\to [R]$ is said to be {\em $\mathcal{A}$-free} if it is $(A^i,\sigma^i)$-free for every $i \geq 1$. The size of $\mathcal{A}$ is the size of the largest induced affine constraint in $\mathcal{A}$.
\item Additionally, $\mathcal{A} = \{(A^1,\sigma^1),
(A^2, \sigma^2), \dots, (A^K,\sigma^K) \}$ is a $W$-light affine system if there exists a basis $\bm{\alpha}=(\alpha_1, \dots, \alpha_r)$ such that $\wt_{\bm{\alpha}}(A^i) \leq W$ for all $i \in [K]$.
\item
A property $\P \subseteq \{\K^n \to [R]\}$ is said to be {\em $K,W$-lightly locally characterized} if it is equivalent to $\mathcal{A}$-freeness for some $W$-light affine system $\mathcal{A}$ whose size is $\leq K$.
\end{itemize}
\end{define}

We recall that Kaufman and Ron \cite{KR06} show that:
\begin{thm}[\cite{KR06}]\label{kaufmanron}
The property $\P_d = \{P: \K^n \to \K : \deg(P)\leq d\}$ is $q^{\left\lceil(d+1)/(q-q/p)\r\rceil},$ $pr\left\lceil(d+1)/(q-q/p)\r\rceil$-lightly locally characterized.\end{thm}

\subsection{Factors and Rank}
Next, we define a polynomial factor which forms the basis for much of higher order Fourier analysis.
\begin{define}[Factor]  A {\em polynomial factor} $\B$ is a sequence of additive polynomials $P_1, \dots, P_C:\K^n \to \T$. We also identify it with the  function $\B:\K^n\to\T^C$ mapping $x$ to $(P_1(x),\ldots,P_C(x))$. An {\em atom} of $\B$ is a preimage $\B^{-1}(y)$ for some $y\in\T^C$. When there is no ambiguity, we will in fact abuse notation and identify an atom of $\B$ with the
common value $\B(x)$ of all $x$ in the atom.

The {\em partition induced by $\B$} is the partition of $\K^n$ given by $\left\{\B^{-1}(y):y\in \T^C\right\}$. The {\em complexity} of $\B$, denoted $|\B|$, is the number of defining polynomials $C$. The {\em order} of $\B$, denoted $\|\B\|$, is the total
number of atoms in $\B$.  The {\em degree} of $\B$ is the maximum additive degree among its defining polynomials $P_1,\ldots,P_C$.
\end{define}

Note that due to Lemma \ref{explicit}, if $\B$ is defined by polynomials $P_1, \dots, P_C$, 
$$\|\B\| = \prod_{i=1}^C p^{\text{depth}(P_i) + 1}$$

\begin{define}[Rank] Let $d \in \N$ and $P:\K^n \rightarrow \T$. Then $\rank_d(P)$ is defined as the smallest integer $k$ such that there exist functions $P_1,\ldots , P_k:\K^n \rightarrow \T$ of additive degree $\leq d-1$ and a function $\Gamma:\T^k \rightarrow \T$ such that $P(x)=\Gamma(P_1(x),\ldots , P_k(x))$. If $d=1$, then the rank is $0$ if $P$ is a constant function and is
$\infty$ otherwise. If $P$ is a polynomial of additive degree $d$, then $\rank(P)=\rank_d(P)$.
\end{define}

\begin{define}[Rank and Regularity of Polynomial Factor]Let $\B$ be a polynomial factor defined by the sequence  $P_1,\ldots , P_c:\K^n \rightarrow \T$ with respective depths $k_1, \dots, k_c$. Then, the rank of $\B$ is  $\min_{(a_1, \dots, a_c)} \rank(\sum_{i=1}^c a_i P_i)$ where the minimum is over $(a_1, \dots, a_c) \in \Z^c$ such that $(a_1 \mod p^{k_1 + 1}, \dots, a_c \mod p^{k_c + 1}) \neq (0, \dots, 0)$ . 

Given a polynomial factor $\B$ and a non decreasing function $r:\Z^+ \to \Z^+$, $\B$ is $r$-regular if $\B$ is of rank at least $r(|\B|)$. 
\end{define}

\begin{define}[Semantic and Syntactic refinement]Let $\B$ and $\B'$ be polynomial factors. A factor $\B'$ is a {\em syntactic refinement} of $\B$, denoted by $\B' \succeq_{syn}\B$ if the set of polynomials defining $\B$ is a subset of the set of polynomials defining $\B'$. It is a {\em semantic refinement}, denoted by $\B' \succeq_{sem}\B$ if for every $x,y \in \K^n$, $\B'(x)=\B'(y)$ implies $\B(x)=\B(y)$. Clearly, a syntactic refinement is also a semantic refinement. 
\end{define}

Our next lemma is the workhorse that allows us to convert any factor into a regular one.
\begin{lem}[Polynomial Regularity Lemma]\label{factorreg}
Let $r: \Z^+ \to \Z^+$ be a non-decreasing function and $d > 0$ be an integer. Then, there is a  function $C_{\ref{factorreg}}^{(r,d)}: \Z^+ \to \Z^+$  such that the following is true. Suppose $\B$ is a factor defined by
polynomials $P_1,\dots, P_C : \K^n \to \T$  of additive degree at most $d$. Then, there is  an $r$-regular factor $\B'$ consisting of  polynomials $Q_1, \dots, Q_{C'}: \K^n \to \T$ of additive degree $\leq d$ such that $\B' \succeq_{sem} \B$ and  $C' \leq C_{\ref{factorreg}}^{(r,d)}(C)$.

Moreover, if $\B$ is itself a refinement of some polynomial factor $\hat{\B}$ that has rank $> (r(C') + C')$, then additionally
$\B'$ will be a syntactic refinement of $\hat{\B}$.
\end{lem}
\begin{proof}
Follows directly from Lemma 2.18 of \cite{BFHHL13} by identifying $\K^n$ with $\F^{rn}$. 
\end{proof}

In fact, the regularization process of \cref{factorreg} can be implemented in time $O(n^{d+1})$ \cite{BHT15}.

\subsection{Gowers norm and the inverse theorem}
\begin{define}
The {\em bias} of a function $f: \K^n \to \C$ is defined as $\bias(f) = \left|\E_{x \in \K^n} f(x)\right|$. For $P: \K^n \to \T$, we use $\bias(P)$ to denote $\bias(e(P))$.
\end{define}

The {\em Gowers norm} of a function measures the bias of its iterated derivative. Precisely:
\begin{define}[Gowers norm]\label{gowers}
Given a function $f: \K^n \to \C$ and an integer $d \geq 1$, the {\em  Gowers norm of order $d$}  for $f$ is given by
$$\|f\|_{U^d} = \left|\E_{h_1,\dots,h_d,x \in \K^n} \left[(\Delta_{h_1} \Delta_{h_2} \cdots \Delta_{h_d}f)(x)\right]\right|^{1/2^d}.$$
If $P: \K^n \to \T$, $\|P\|_{U^d}$ denotes $\|e(P)\|_{U^d}$. 
\end{define}
Note that as $\|f\|_{U^1}= \bias(f)$ the Gowers norm of order $1$ is only a semi-norm. However for $d>1$, it is not difficult to show that
$\|\cdot\|_{U^d}$ is indeed a norm.

There is a tight connection between additive polynomials and  Gowers norms. In one direction, it is a straightforward consequence of the monotonicity of the Gowers norm ($\|f\|_{U^d} \leq \|f\|_{U^{d+1}}$) and invariance of the Gowers norm with respect to modulation by lower degree polynomials ($\|f\|_{U^{d+1}} = \|f\cdot e(P)\|_{U^{d+1}}$ for polynomials $P$ of additive degree $\leq d$) that if $f: \K^n \to \C$ is {\em $\delta$-correlated} with a polynomial $P$ of additive degree $\leq d$, meaning 
$$|\E_{x} f(x) e(-P(x))| \geq \delta$$ 
for some $\delta > 0$, then 
$$\|f\|_{U^{d+1}} \geq \delta.$$

In the other direction, we have the following ``Inverse theorem for the Gowers norm''.
\begin{thm}[Theorem 1.11 of \cite{TZ12}]\label{inverse}
Suppose $\delta > 0$ and $d \geq 1$ is an integer. There exists an $\eps = \eps_{\ref{inverse}}(\delta,d)$ such that the following
holds. For every function $f: \K^n \to \C$ with $\|f\|_\infty \leq 1$ and $\|f\|_{U^{d+1}} \geq \delta$, there exists a polynomial $P:
\K^n \to \T$ of additive degree $\leq d$ that is $\eps$-correlated with $f$, meaning
$$\left|\E_{x \in \K^n} f(x) e(-P(x)) \right| \geq \eps.$$
\end{thm}

We can be more explicit when $f = e(P)$ for an additive polynomial $P$.
\begin{thm}[Theorem 1.20 of \cite{TZ12}]\label{invpoly}
Suppose $\delta > 0$ and $d \geq 1$ is an integer. There exists an $r = r_{\ref{invpoly}}(\delta,d)$ such that the following
holds. If a polynomial $P: \K^n \to \T$ with additive degree $d$ satisfies $\|P\|_{U^{d}} \geq \delta$, then $\rank(P) \leq r$.

\end{thm}

\section{New Tools}

\subsection{Equidistribution of regular factors}

Our results in this section imply that a regular polynomial factor is ``as random as possible'', subject to the additive degree and depth bounds of its defining polynomials. Let us start with the following simple observation.
\begin{lem}\label{atomsize}
Given $\eps > 0$, let $\B$ be a polynomial factor of degree $d>0$, complexity $C$ and rank $r_{\ref{atomsize}}(d,\eps)$, defined by a sequence of additive polynomials $P_1, \dots, P_C : \K^n \to \T$ having respective depths $k_1, \dots, k_C$. Suppose $\alpha = (\alpha_1, \dots, \alpha_C) \in \U_{k_1 + 1} \times \cdots \times \U_{k_C+1}$. Then:
$$
\Pr_x[\B(x) = \alpha] = \frac{1}{\|\B\|} \pm \eps.
$$
\end{lem}
\begin{proof}
This is standard. See for example Lemma 3.2 of \cite{BFHHL13}.
\end{proof}

In our applications though, we will often need not just $\B(x)$ to be nearly uniformly distributed but the tuple $(\B(x) : x \in H)$ for a set $H\subseteq \K^n$ to be nearly uniformly distributed. In particular, we consider the case when $H$ is an affine subspace of $\K^n$. The following lemma is key.
\begin{lem}[Near orthogonality]\label{lem:equiaffine} Let $A=(L_1,\ldots , L_m)$ be a weight-closed affine constraint of bounded size on $\ell$ variables.  Suppose $\B$ is a polynomial factor of degree $d$ and rank $\geq r^{(\ref{invpoly})}(d,\de)$, defined by the sequence of additive polynomials $P_1, \dots, P_c: \K^n \to \T$. Let $\Lambda=(\lambda_{ij})_{i \in [c], j \in [m]}$ be a tuple of integers. Define: $$P_{\Lambda}(x_1,\ldots ,x_{k})=\sum_{i \in [c], j \in [m]}\lambda_{ij}P_i(L_j(x_1,\ldots , x_{\ell})).$$ Then one of the following is true.
\begin{enumerate}
\item For every $i \in [c]$, it holds that $\sum_{j \in [m]}\lambda_{ij}Q_i(L_j(\cdot))\equiv 0$ for all polynomials $Q_i:\K^n \to \T$ with the same additive degree and depth as $P_i$. Clearly, this implies $P_{\Lambda}\equiv 0$.
\item $P_{\Lambda}\not \equiv 0$. Moreover, $\bias(P_\Lambda)\leq \de$.
\end{enumerate}
\end{lem}
\begin{proof}
For $j \in [m]$, let $(w_{j,1}, \ldots , w_{j,\ell}) \in \K^{\ell}$ denote the affine form given by $L_j$. Note that $w_{j,1}=1$. 

Suppose $\bm{\alpha}=(\alpha_1, \dots, \alpha_r)$ is the basis with respect to which the affine forms are weight-closed.
For each $i$, we do the following. If for some $j$, we have\footnote{Here, $\deg(\cdot)$ refers to the additive degree.} $\wt_{\bm{\alpha}}(L_j) > \deg(\lambda_{i,j}P_i)$, $\lambda_{i,j}\neq 0$, then using \cref{dual}, $L_j(x_1, \dots, x_\ell) = x_1 + \sum_{i=2}^\ell \l(\sum_{k=1}^r u_{i,k} \cdot \beta_k \r) x_i$ where $\bm{\beta}$ is the dual basis to $\bm{\alpha}$, each $u_{i,k} \in [0,p-1]$ and $\sum_{i,k} u_{i,k} > \deg(\lambda_{i,j} P_i)$. 
Using \cref{defadd}, we can replace $\lambda_{i,j}P_i(L_j)$ by a $\Z$-linear combination of $P_i(L_{j'})$ where $L_{j'}\preceq_{\bm{\alpha}} L_j$ until no such $j$ exists. This is where we use the fact that the affine constraint is weight-closed. Suppose the new coefficients are denoted by $(\lambda_{i,j}')$. If the $\lambda_{i,j}'$ are all zero, then for every $i \in [c]$ individually, $\sum_{j \in [m]} P_i(L_j(x_1,\ldots , x_\ell)) \equiv 0$. Indeed, $\sum_{j \in [m]} Q_i(L_j(x_1,\ldots , x_\ell)) \equiv 0$ for any $Q_i$ with the same additive degree and depth, as the transformation from $\lambda_{i,j}$ to $\lambda'_{i,j}$ did not use  any other information about $P_i$.

Else some $\lambda_{i,j}' \neq 0$. Also, $\wt_{\bm{\alpha}}(L_j) \leq \deg(\lambda_{i,j}'P_i)$. Then we show the second part of the lemma, that is $\l|\E[e(P_{\Lambda}(x_1,\ldots , x_{k})]\r|\leq \de$.

Suppose without loss of generality that the following is true.
\begin{itemize}
\item $\lambda_{i,1}' \neq 0$ for some $i \in [C]$.
\item $L_1$ is maximal in the sense that for every $j \neq 1$, either $\lambda_{i,j}'=0$ for all $i \in [C]$ or $\wt_{\bm{\alpha}}(w_{j,s})<\wt_{\bm{\alpha}}(w_{1,s})$ for some $s \in [\ell]$.
\end{itemize}

For $a=(a_1,\ldots , a_{\ell})\in \K^{\ell}$ and $y \in \K^n$ and $P:\K^n \to \T$, define $$\overline{D}_{a, y}P(x_1,\ldots , x_{\ell})=P(x_1+a_1y,\ldots , x_{\ell}+a_{\ell}y)-P(x_1,\ldots , x_{\ell}).$$

Then $$\overline{D}_{a, y}(P_i\circ L_j)(x_1,\ldots , x_{\ell})=(D_{L_j(a)y}P_i)(L_j(x_1,\ldots , x_{\ell})).$$

Let $\Delta=\wt_\alpha(L_1) \leq d$. Define $a_1,\ldots , a_{\De}$ be the set of vectors of the form $(-w,0,\ldots , 1,0, \ldots , 0)$ where $1$ is in the $i$th coordinate for $i \in [2,\ell]$ and for all $w \in \K$ satisfying $0 \leq \wt_{\bm{\alpha}}(w)  < \wt_{\bm{\alpha}}(w_{1,i})$. Note that $\langle L_1,a_k \rangle \neq 0$ for $k \in [\De]$ but for any $j>1$ there exists some $k \in [\De]$ such that $\langle L_j,a_k \rangle=0$. Thus, $$\E_{y_1,\ldots , y_{\De},x_1,\ldots , x_{\ell}}\l[e\l((\overline{D}_{a_{\De},y_{\De}}\ldots \overline{D}_{a_1.y_1}P_{\Lambda})(x_1,\ldots , x_{\ell})\r)\r]=\l\|\sum_{i=1}^C \lambda_{i,1}'P_i\r\|^{2^{\De}}_{U^{\De}}.$$ 
The rest of the analysis is same as Theorem 3.3 in \cite{BFHHL13} and we skip it here.
\end{proof}
We can now use \cref{lem:equiaffine} to prove our result on equidistribution of regular factors over affine subspaces of $\K^n$.
\begin{thm}\label{unifsubspace}
Let $\eps > 0$
Let $\B$ be a polynomial factor defined by polynomials $P_1, \dots, P_c: \K^n \to \T$ with respective additive degrees $d_1, \dots, d_c \in \Z^+$ and depths $k_1, \dots, k_c \in \Z^{\geq 0}$.  Suppose $\B$ has rank at least $r^{(\ref{invpoly})}(d,\eps)$ where $d = \max(d_1, \dots, d_c)$. Let $A = (L_1, \dots, L_m)$ be a weight-closed affine constraint. For every $i \in [c]$, define $\Lambda_i$ to be the set of tuples $(\lambda_{1}, \dots, \lambda_{m}) \in [0, p^{k_i + 1}-1]$ such that $\sum_{j=1}^m \lambda_{j} Q_i(L_j(\cdot)) \equiv 0$ for all polynomials $Q_i$ with the same additive degree and depth as $P_i$. 

Consider $(\alpha_{i,j}: i \in [c], j \in [m]) \in \T^{cm}$ such that for every $i \in [c]$ and for every $(\lambda_1, \dots, \lambda_m) \in \Lambda_i$, $\sum_{j=1}^m \lambda_j \alpha_{i,j} = 0$. Then:
$$\Pr_{x_1, \dots, x_\ell \in \K^n}[\B(L_j(x_1, \dots, x_\ell)) = (\alpha_{1, j}, \dots, \alpha_{c, j})~ \forall j \in [m]] = \frac{\prod_{i=1}^c |\Lambda_i|}{\|\B\|^m} \pm \eps
$$
\end{thm}
\begin{proof}
\begin{align*}
&\Pr_{x_1, \dots, x_\ell \in \K^n}[\B(L_j(x_1, \dots, x_\ell)) = (\alpha_{1, j}, \dots, \alpha_{c, j})~ \forall j \in [m]]\\
&=\E_{x_1, \dots, x_\ell} \left[\prod_{i,j} \frac{1}{p^{k_i + 1}} \sum_{\lambda_{i,j}=0}^{p^{k_i + 1}-1} e(\lambda_{i,j} (P_{i}(L_j(x_1, \dots, x_\ell)) - \alpha_{i,j}))\right]\\
&= \left(\prod_{i}p^{-(k_i+1)}\right)^m \sum_{(\lambda_{i,j}) \atop \in \prod_{i, j} [0,p^{k_i + 1}-1]} e\left(-\sum_{i,j} \lambda_{i,j}\alpha_{i,j}\right) \E \left[e\left(\sum_{i,j} \lambda_{i,j}P_{i}(L_j(x_1,
    \dots, x_\ell))\right)\right]\\
&= p^{-m \sum_{i=1}^c(k_i+1)} \cdot \left(\prod_{i=1}^c |\Lambda_i|~ \pm~  \eps p^{m \sum_{i=1}^c(k_i+1)}\right)
\end{align*}
The last line is due to the observation that from \cref{lem:equiaffine}, $\sum_{i=1}^c \sum_{j=1}^m \lambda_{i,j} P_i(L_j(x_1, \dots, x_\ell)) \equiv 0$ if and only if for every $i \in [c]$, $(\lambda_{i,1}, \dots, \lambda_{i,m}) \in \Lambda_i \pmod {p^{k_i + 1}}$. So, $\sum_{i,j}\lambda_{i,j} P_i(L_j(\cdot))$ is identically $0$ for $\prod_i |\Lambda_i|$ many tuples $(\lambda_{i,j})$ and for those tupes, $\sum_{i,j} \lambda_{i,j} \alpha_{i,j} = 0$ also.

\end{proof}

Note that in \cref{unifsubspace}, if $\eps$ is a constant, $m$ needs to be bounded for the claim to be non-trivial, which in turn requires that the affine forms in $L$ be of bounded weight.

\subsection{Preservation of Locally Characterized Properties}

\begin{thm}\label{proppreserve}
Let $\P \subset \{\K^n \to \K\}$ be a $K,W$-lightly locally characterized property. For an integer $d$, suppose $P_1, \dots, P_c: \K^n \to \T$ are polynomials of additive degree $\leq d$, forming a factor of rank $> r_{\ref{proppreserve}}(d,K),$ and $\Gamma: \T^c \to \K$ is a  function such that $F: \K^n \to \K$ defined by $F(x) = \Gamma(P_1(x), \dots, P_c(x))$ satisfies $\P$. 

For every collection of additive polynomials $Q_1, \dots, Q_c: \K^n \to \T$ with $\deg(Q_i) \leq \deg(P_i)$ and $\depth(Q_i) \leq \depth(P_i)$ for all $i \in [c]$, if $G: \K^n \to \K$ is defined by $G(x) = \Gamma(Q_1(x), \dots, Q_c(x))$, then $G \in \P$ too.
\end{thm}

\begin{proof}
For the sake of contradiction, suppose $G \notin \P$. Then, for a weight-closed affine constraint consisting of $K'$ linear forms $L_1, \dots, L_{K'}$, there exist $x_1, \dots, x_\ell$ such that $(G(L_1(x_1, \dots, x_\ell)), \dots,$ $G(L_{K'}(x_1, \dots, x_\ell)))$ which form a witness to $G \not \in \P$. Note that $K'$ is a function of only $K$ and $W$ because the affine forms characterizing $\P$ can be made weight $\leq W$ by a choice of basis for $\K$ over $\F$ and then completed into a weight-closed constraint. So, there exists $x_1, \dots, x_\ell \in \K^n$ such that the tuple $B = (Q_i(L_j(x_1, \dots, x_\ell)): j \in [K'], i \in [c]) \in \T^{cK'}$ is a proof of the fact that $G \not \in \P$.

Now we argue that 	there exist $x'_1, \dots, x'_\ell$ such that $(P_i(L_j(x'_1, \dots, x'_\ell)):  i \in [c], j \in [K])$ equals $B$, thus showing that $F \not \in P$, a contradiction. Notice that $B$ satisfies the conditions required of $\alpha$ in \cref{unifsubspace}. So by \cref{unifsubspace},
$$
\Pr_{x'_1, \dots, x_\ell'}\left[\left(P_i(L_j(x'_1, \dots, x'_\ell): i \in [c], j \in [K]\right) = B\right] > 0
$$
if the rank of the factor formed by $P_1, \dots, P_c$ is more than $r^{(\ref{invpoly})}\left(d, \frac{1}{2\|\B\|^{K}}\right)$, where $\|\B\| = p^{\sum_{i=1}^c (\depth(P_i) + 1)}$. 
\end{proof}

In our applications, we will use \cref{proppreserve} for the property of having bounded degree, which is lightly locally characterized by \cref{kaufmanron}.

\section{List decoding of RM codes}

We state the following corollary which we need in the proof to follow. We only state a special case of it which is enough.

\begin{cor}[Corollary 3.3 of \cite{BL14}]\label{cor:pseudorandom} Let $g:K \rightarrow K$, $\eps>0$. Then there exist $c \leq 1/\eps^2$ functions $h_1, h_2,\ldots , h_c \in \RM_{\K}(n,d)$ such that for every $f \in \RM_{\K}(n,d)$, there is a function $\Gamma_f:\K^c \rightarrow \K$ such that $$\Pr_x[\Gamma_f(h_1(x), \ldots , h_c(x))=f(x)] \geq \Pr_x[g(x)=f(x)] - \eps.$$
\end{cor}

\textbf{Theorem~\ref{thm:main} (Restated).} Let $\K=\F_q$ be an arbitrary finite field. Let $\eps>0$ and $d,n \in \N$. Then, $$\ell_{\K}(d,n,\de_{\K}(d)-\eps) \leq c_{q,d,\eps}.$$

\begin{proof}
We follow the proof structure in \cite{BL14}.
Let $g:\K^n \rightarrow \K$ be a received word. Suppose $\Pr[g(x)=f(x)] \ge 1- \de_{\K}(d)+\eps$. Apply Corollary~\ref{cor:pseudorandom} with approximation parameter $\eps/2$ gives $\calH_0=\{h_1, \ldots , h_c\} \subseteq \RM_{\K}(n,d)$, $c \leq 4/\eps^2$ such that, for every $f \in \RM_{\K}(n,d)$, there is a function $\Gamma_f:\K^c \rightarrow \K$ satisfying
$$\Pr[\Gamma_f(h_1(x),h_2(x),\ldots , h_c(x))=f(x)] \geq \Pr[g(x)=f(x)] - \eps/2 \ge 1- \de_{\K}(d)+\eps/2.$$

Let $\alpha_1, \alpha_2, \dots, \alpha_r$ be an arbitrary basis for $\K$ over $\F$. Let $\de(d):=\de_{\K}(d)$.
By Proposition~\ref{dual}, $$\Pr[\Gamma_f'(\Tr(\alpha_i h_j(x)):1 \le i \le r, 1 \le j \le c )=F(\Tr(\alpha_if(x)):1 \le i \le r)] \ge  1-\de(d) + \eps/2,$$
 where $\Gamma_f':\F^{rc} \to \K$ and $F:\F^r \to \K$. From here onwards, we identify $\F$ with $\U_1$.
Let $\calH=\{\Tr(\alpha_i h_j(x)):1 \le i \le r, 1 \le j \le c \}$ and $\calH_F=\{\Tr(\alpha_if(x)):1 \le i \le r)\}$.

Let $r_1, r_2:\N \rightarrow \N$ be two non decreasing functions to be specified later, and let $C_{r,d}^{(\ref{factorreg})}$ be as given in Lemma~\ref{factorreg}. We will require that for all $m \ge 1$, \begin{equation}\label{eq:r1r2}  r_1(m)\geq r_2(C_{r_2,d}^{(\ref{factorreg})}(m+1))+C_{r_2,d}^{(\ref{factorreg})}(m+1)+1.
\end{equation}

As a first step, we $r_1$-regularize $\calH$ by Lemma~\ref{factorreg}. This gives an $r_1$-regular factor $\B'$ of degree at most $d$, defined by polynomials $H_1,\ldots , H_{c}:\K^n \rightarrow \T$, $c' \leq C_{r_1,d}^{(\ref{factorreg})}(cr)$ and $\rank(\B')  \geq  r_1(c')$. We denote $\calH'=\{H_1,\ldots , H_{c'}\}$. Let $\depth(H_i)=k_i$ for $i \in [c']$. Let $G_f:\otimes_{i=1}^{c'} \U_{k_i+1} \rightarrow \U_1$ be defined such that
$$
\Gamma_f(h_1(x),\ldots , h_c(x))=G_f(h_1'(x),\ldots , h_{c'}'(x)).
$$

Next, we will show that $f$ is measurable with respect to $\calH'$ and this would upper bound the number of such polynomials by $c'(q,d,\eps)$ independent on $n$.

Fix such a polynomial $f$. Call $F_i=\Tr(\alpha_i f)$. Appealing again to Lemma~\ref{factorreg}, we $r_2$-regularize $\B_f:=\B' \bigcup \calH_F$. We get an $r_2$-regular factor $\B'' \succeq_{syn} \B'$ defined by the collection $\calH''=\{H_1,\ldots , H_{c'},H'_1,\ldots , H'_{c''}\}$. Note that it is a syntactic refinement of $\B'$ as by our choice of $r_1$, $$\rank(\B') \geq r_1(c')  \geq  r_2(C_{r_2,d}^{(\ref{factorreg})}(c'+1))+C_{r_2,d}^{(\ref{factorreg})}(c'+1)+1 \geq r_2(|\B''|)+|\B''|+1.$$
We will choose $r_2$ such that for all $m \ge 1$,
\begin{equation}\label{eq:r2atom}
r_2(m) = \max\l(r_d^{(\ref{atomsize})}\l(\frac{\eps/4}{\l(p^{\lfloor\frac{d-1}{p-1}\rfloor+1}\r)^m}\r),r^{(\ref{proppreserve})}_d(m)\r).
\end{equation}
Since each $F_i$ is measurable with respect to $\B''$, there exists $F':S \rightarrow \U_1$ such that
$$
f(x)=F'(H_1(x),\ldots , H_{c'}(x), H'_1(x),\ldots , H'_{c''}(x)).
$$
Summing up, we have $$\Pr[G(H_1(x),H_2(x),\ldots , H_{c'}(x))=F'(H_1(x),\ldots , H_{c'}(x), H'_1(x),\ldots , H'_{c''}(x))] \ge 1-\de(d)+ \eps/2.$$
We next show that we can have each polynomial in the factor have a disjoint set of inputs. This would simplify the analysis considerably.

\begin{claim}\label{clm:tild}Let $x^i, y^j$, $i \in [c'], j \in [c'']$ be pairwise disjoint sets of $n \in \N$ variables each. Let $n' = n(c'+c'')$. Let $\tilde{f}:\K^{n'} \rightarrow \K$ and $\tilde{g}:\K^{n'} \rightarrow \K$ be defined as   $$\tilde{f}(x)=F(H_1(x^1),\ldots , H_{c'}(x^{c'}), H'_1(y^1), \ldots , H'_{c''}(y^{c''}))$$ and $$\tilde{g}(x)=G(H_1'(x^1), \ldots , H_{c'}(x^{c'})).$$
Then $\deg(\tilde{f}) \le d$ and
$$
\l|\Pr_{x \in \F^{n'}}[\tilde{f}(x)=\tilde{g}(x)] - \Pr_{x \in \F^n}[f(x)=G_f(h_1'(x),h_2'(x),\ldots , h_c'(x))]\r|\leq  \eps/4.
$$
\end{claim}

\begin{proof}
The bound $\deg(\tilde{f}) \le \deg(f) \le d$ follows from Lemma~\ref{proppreserve} since $r_2(|\calH''|) \ge r^{(\ref{proppreserve})}_d(|\calH''|)$. To establish the bound on $\Pr[\tilde{f}=\tilde{g}]$, for each $s \in S$ let
$$
p_1(s) = \Pr_{x \in \F^n}[(h'_1(x),\ldots,h'_{c'}(x),h''_1(x),\ldots,h''_{c''}(x))=s].
$$
Applying Lemma~\ref{atomsize} and since our choice of $r_2$ satisfies $\rank(\calH'') \ge r_d^{(\ref{atomsize})}(\eps/4|S|)$, we have that $p_1$ is nearly uniform over $S$,
$$
p_1(s) = \frac{1 \pm \eps/4}{|S|}.
$$
Similarly, let
$$
p_2(s) = \Pr_{x^1,\ldots,x^{c'}, y^1,\ldots,y^{c''} \in \F^n}[(h'_1(x^1),\ldots,h'_{c'}(x^{c'}),h''_1(y^1),\ldots,h''_{c''}(y^{c''}))=s].
$$
Note that the rank of the collection of polynomials $\{h'_1(x^1),\ldots,h'_{c'}(x^{c'}),h''_1(y^1),\ldots,h''_{c''}(y^{c''})\}$ defined over $\F^{n'}$ cannot be lower than that of $\calH''$. Applying Lemma~\ref{atomsize} again gives
$$
p_2(s) = \frac{1 \pm \eps/4}{|S|}.
$$
For $s \in S$, let $s' \in \otimes_{i=1}^{c'}\U_{k_i+1}$ be the restriction of $s$ to first $c'$ coordinates, that is, $s'=(s_1,\ldots ,s_{c'})$. Thus
\begin{align*}
\Pr_{x \in \F^{n'}}[\tilde{f}(x)=\tilde{g}(x)] &= \sum_{s \in S} p_2(s) 1_{F(s)=G_f(s')} \\
&= \sum_{s \in S} p_1(s) 1_{F(s)=G_f(s')} \pm \eps/4 \\
&= \Pr_{x \in \F^n}[f(x)=G_f(h_1'(x),h_2'(x),\ldots , h_c'(x))] \pm \eps/4.
\end{align*}
\end{proof}

So, we obtain that
$$
\Pr_{x \in \F^{n'}}[\tilde{f}(x)=\tilde{g}(x)] \ge \Pr_{x \in \F^n} [f(x) = G_f(h'_1(x),\ldots,h'_{c'}(x))] - \eps/4 \ge 1 - \de(d)+\eps/4.
$$
Next, we need the following variant of the Schwartz-Zippel lemma from \cite{BL14}.
\begin{claim}\label{clm:sz1}
Let $d,n_1,n_2 \in \N$. Let $f_1:\K^{n_1+n_2} \rightarrow \K$ and $f_2:\K^{n_1} \rightarrow \K$ be such that $\deg(f_1)\leq d$ and $$\Pr[f_1(x_1,\ldots , x_{n_1+n_2})=f_2(x_1,\ldots , x_{n_1})]>1-\de(d)$$
Then, $f_1$ does not depend on $x_{n_1+1}, \ldots , x_{n_1+n_2}$.
\end{claim}

With claim \ref{clm:sz1} applied to $f_1=\tilde{f}, f_2=\tilde{g}, n_1=n c', n_2= nc''$. We obtain that $\tilde{f}$ does not depend on $y^1,\ldots,y^{c''}$. Hence,
$$
\tilde{f}(x^1,\ldots,x^{c'},y^1,\ldots,y^{c''})=F(H_1'(x^1),\ldots , H_{c'}'(x^{c'}), C_1 ,\ldots , C_{c''})$$
where $C_j=H''_j(0)$ for $j \in [c'']$. If we substitute $x^1=\ldots=x^{c'}=x$ we get that
$$
f(x)=F(H'_1(x),\ldots,H'_{c'}(x),H''_1(x),\ldots,H''_{c''}(x)) =
F(H_1'(x),\ldots , H_{c'}'(x), C_1, \ldots , C_{c''}),
$$
which shows that $f$ is measurable with respect to $\calH'$, as claimed.

\end{proof}

\section{Polynomial decomposition}

\begin{define}Given $k \in \N$ and $\De=(\De_1,\ldots , \De_k) \in \N^k$ and a function $\Gamma:\K^k \to \K$, a function $P:\K^n \to \K$ is $(k,\De,\Gamma)$-structured if there exist polynomials $P_1,\ldots , P_k:\K^n \to \K$ with $\deg(P_i) \le \De_i$ such that for $x \in \K^n$, we have $$P(x)=\Gamma(P_1(x),\ldots , P_k(x)).$$ The polynomials $P_1,\ldots , P_k$ form a $(k,\De,\Gamma)$-decomposition.
\end{define}

The main result we prove is the following.
\begin{thm}\label{thm:polydeco}Let $k \in \N$. For every $\De = (\De_1, \ldots , \De_k) \in \N^k$ and every function $\Gamma: \K^k  \to \K$, there is a randomized algorithm $A$ that on input $P: \K^n \to \K$ of degree $d$, runs in time $\poly_{q,k,\De}(n^{d+1})$ and outputs a $(k, \De, \Gamma)$-decomposition of $P$ if one exists while otherwise returning $NO$.
\end{thm}

We first show that the notion of rank is robust to hyperplane restrictions over nonprime fields. More precisely, we have the following.
\begin{lem}\label{lem:hyperplanerank}Let $P:\K^n \to \T$ be an additive polynomial such that $\rank(P) \geq r$. Let $H$ be a hyperplane in $\K^n$. Then the restriction of $P$ to $H$ has rank at least $r-q$.
\end{lem}
\begin{proof}
Without loss of generality, let $H$ be defined by $x_1=0$. Let $P':\K^{n-1}\rightarrow \T$ be the restriction of $P$ defined by $P'(y)=P(0y)$. Let $\pi:\K^{n}\rightarrow \K^{n-1}$ be the map $\pi(x_1x_2\ldots x_n)=x_2\ldots x_n$. Let $P'':\K^n \rightarrow \T$ be defined by $P''(x)=P(x)-P'\circ\pi$. Then $P''(x)=0$ for $x \in H$. For $i \in \K\setminus\{0\}$, let $h_i=(i,0,\ldots ,0)$. Then, for $y \in H$, define $R_j:\K^n \to \T$ by $$R_j(y)=P''(y+h_j)=(D_{h_j}P'')(y).$$
Note that $\deg(R_j) \leq d-1$. Now, since $P(x)=P''(x)+P'\circ\pi(x)$, we have $$P(x)=\Gamma(P'\circ \pi, x_1, \{R_y(x):y \in \F\}).$$ 
Now, if $\rank(P') \leq r$, then $\rank(P'\circ \pi) \leq r$ and hence $\rank(P) \leq r+q$. This finishes the proof.
\end{proof}

We now start with the proof of Theorem~\ref{thm:polydeco}.
\begin{proof}
Let $R_1:\N \to \N$ be defined as $R_1(m)=R_2(c_{\ref{factorreg}}^{(R_1,d)}(m+k))+c_{\ref{factorreg}}^{(R_1,d)}(m+k)+q$ where $R_2:\N \to \N$ will be specified later.

Let $\alpha_1, \dots, \alpha_r$ be an arbitrary basis for $\K$ over $\F$.
By Proposition~\ref{dual}, $P(x)=\sum_i \beta_i \Tr(\alpha_iP(x))$ for the dual basis $\beta_1, \dots, \beta_r$. Set $f_i(x)=\Tr(\alpha_iP(x))$. Identifying $\F$ with $\U_1$ we treat $f_i:\K^n \to \T$. Regularize $\{f_1,\ldots , f_r\}$ using the algorithm of \cite{BHT15} to find $R_1$-regular $\B=\{g_1,\ldots , g_C:\K^n \to \T\}$ where $C \le c_{\ref{factorreg}}^{(R_1,d)}(r)$. So, $f_i(x)=G_{i}(g_1(x),\ldots , g_C(x))$ and $P(x)=\sum_i \alpha_i G_{i}(g_1(x),\ldots , g_C(x)).$ Thus, if $n \le Cd$, then we are done by a brute force search.

Else, $n >Cd$. For each $g_i$, pick a monomial $m_i$ with degree $\deg(P_i)$. Then there is $i_0 \in [n]$ such that $x_{i_0}$ does not appear in any $g_i$. Set $g_i':=g_i|x_{i_0}=0$. Let $\B'$ be the factor defined by the $g_i's$. Note that $\deg(g_i')=\deg(g_i)$ and $\depth(g_i')=\depth(g_i)$. Also, by Lemma~\ref{lem:hyperplanerank}, $\B'$ is $R_1-q$-regular.

Now, using recursion, we solve the problem on $n-1$ variables. That is, decide if for $P':=P|x_{i_0}=0$ is $(k,\De,\Gamma)$-structured. If $P'$ is not, then $P$ is not either, so we are done. Else, suppose the algorithm does not output NO.

Say $$P'(x)=\Gamma(S_1(x),\ldots , S_k(x))=\Gamma'(\Tr(\alpha_jS_i(x)):i \in [k], j \in [r]),$$ where $$\Gamma'(a_{ij}:i \in [k],j\in [r])=\Gamma(\sum_j \beta_i a_{ij}: i \in [k]).$$ Note that while $\Gamma:\K^k \to \K$, we have $\Gamma':\F^{kr}\to \K$. Let $\B_1$ be the factor formed by $\{\Tr(\alpha_j S_i)\}$. Via the algorithm of \cite{BHT15}, regularize $\B' \cup \B_1$ using $R_2:\N \to \N$ and we get a syntactic refinement $\B' \cup \B_1'$ by the choice of $R_1$. Let $\B_1'=\{s_1',\ldots , s_D'\}.$ where $$\Tr(\alpha_j S_i)=G_{ij}(g_i',s_j': i \in [C], j \in [D]).$$ Choose $R_2$ large enough such that the map induced by $\B' \cup \B_1'$ is surjective. Now, fix any $\ell \in [r]$. Then, $$\Tr(\alpha_{\ell}P')=G_{\ell}(g_1',\ldots , g_C')=F_{\ell}(G_{ij}(g_i',s_j')),$$
where $F_{\ell}=\Tr(\alpha_{\ell}\Gamma')$. Thus, for $a_1,\ldots , a_C, b_1,\ldots  b_D \in \F$, $$G_{\ell}(a_1,\ldots , a_C)=F_{\ell}(G_{ij}(a_1,\ldots , b_D): i \in [C], j \in [D]).$$  Substituting, $a_i=g_i(x)$ and $b_j=0$ we have $$\Tr(\alpha_{\ell}P)=G_{\ell}(g_1,\ldots , g_C)=F_{\ell}(G_{ij}(g_i,0)).$$ Now, $$\Tr(\alpha_{\ell}P)=\Tr(\alpha_{\ell}\Gamma(Q_i:i \in [k])),$$ where $Q_i(x)=\sum_{j=1}^r \alpha_j G_{ij}(g_i',\ldots , 0)$.

Since, this is true for all $\ell \in [r]$, we have $$P(x)=\Gamma(Q_1(x),\ldots, Q_k(x)).$$ where $Q_i$ is defined as above.
This finishes the proof.
\end{proof}

\bibliographystyle{alpha}
\bibliography{testing}

\end{document}